\documentclass[final,3p,times]{elsarticle}
\usepackage{amsmath}
\usepackage{amssymb}
\usepackage{wasysym}
\usepackage{mathrsfs}
{\bfseries}{\itshape}
\newtheorem{proposition}{Proposition}
\newtheorem{lemma}{Lemma}
\newproof{proof}{Proof}
\newdefinition{definition}{Definition}
\newtheorem{theorem}{Theorem}

  \newtheorem{corollary}{Corollary}
\journal{arXiv}

\newcommand{\boxi}{\operatorname{box}}
\newcommand{\cub}{\operatorname{cub}}
\newcommand{\dime}{\operatorname{dim}}
\newcommand{\idime}{\operatorname{idim}}
\newcommand{\ch}{\operatorname{ch}}
\begin{document}
\begin{frontmatter}
\title{Sublinear Approximation Algorithms for Boxicity and Related Problems\footnote{A preliminary version of this work appeared in IPEC 2012.}}
\author[abh]{Abhijin Adiga}
\ead{abhijin@vbi.vt.edu}
\author[jas]{Jasine Babu}
\ead{jasinekb@gmail.com}
\author[sun]{L. Sunil Chandran}
\ead{sunil@csa.iisc.ernet.in}
\address[abh]{Virginia Tech., USA.}
\address[jas]{University of Haifa, Israel.}
\address[sun]{Indian Institute of Science, Bangalore, India.}
\begin{abstract}
Boxicity of a graph $G(V,$ $E)$ is the minimum  integer $k$ such that $G$ can be represented as the intersection 
graph of axis parallel boxes in $\mathbb{R}^k$. 
Cubicity is a variant of boxicity, where the axis parallel boxes in the intersection representation are restricted to be of unit length sides.
Deciding whether boxicity (resp. cubicity) of a graph is at most $k$ is NP-hard, even for $k=2$ or $3$. 
Computing these parameters is inapproximable within $O(n^{1 - \epsilon})$-factor, for any $\epsilon >0$ in polynomial time 
unless $\text{NP}=\text{ZPP}$, even for many simple graph classes. 

In this paper, we give a polynomial time $\kappa(n)$ factor approximation algorithm for computing boxicity and a $\kappa(n)\lceil \log \log n\rceil$ factor 
approximation algorithm for computing the cubicity, where $\kappa(n) =2\left\lceil\frac{n\sqrt{\log \log n}}{\sqrt{\log n}}\right\rceil$.
These $o(n)$ factor approximation algorithms also produce the corresponding box (resp. cube) representations. 
As a special case, this resolves the question posed by Spinrad \cite{Spinrad2003} about polynomial time construction of $o(n)$ 
dimensional box representations for boxicity $2$ graphs.
Other consequences of our approximation algorithm include $O(\kappa(n))$ factor approximation algorithms for computing the 
following parameters: the partial order dimension (poset dimension) of finite posets, the interval dimension of finite posets, minimum chain cover of bipartite graphs,
Ferrers dimension of digraphs and threshold dimension of split graphs and co-bipartite graphs. 
Each of these parameters is inapproximable within an $O(n^{1 - \epsilon})$-factor, for any $\epsilon >0$ in polynomial time 
unless $\text{NP}=\text{ZPP}$ and the algorithms we derive seem to be the first $o(n)$ factor approximation algorithms known 
for all these problems. We note that obtaining a $o(n)$ factor approximation for poset dimension was also mentioned as an open problem by 
Felsner et al.~\cite{Felsner2015}.
\end{abstract}
\end{frontmatter}
\section{Introduction}\label{sectIntro}
Let $G(V$, $E)$ be a graph. If  $I_1$, $I_2$, $\cdots$, $I_k$ are (unit) interval graphs on the vertex set $V$ such that 
$E(G)=E(I_1) \cap E(I_2) \cap \cdots \cap E(I_k)$, then $\{I_1$, $I_2$, $\cdots$, $I_k\}$ is called a box (cube) representation of 
$G$ of dimension $k$. Boxicity (cubicity) of a non-complete graph $G$, denoted by $\boxi(G)$ (respectively $\cub(G)$), is defined as the minimum 
integer $k$ such that $G$ has a box (cube) representation of dimension $k$. For a complete graph, it is defined to be zero. 
Equivalently, boxicity (cubicity) is the minimum integer $k$ such that $G$ can be represented as the intersection graph of axis 
parallel boxes (cubes) in $\mathbb{R}^k$. Boxicity was introduced by Roberts \cite{Rob1} in 1969 for modeling problems in social 
sciences and ecology. Some well known NP-hard problems like the max-clique 
problem are polynomial time solvable, if low dimensional box representations are known \cite{Rosgen}.

For any graph $G$ on $n$ vertices, $\boxi(G) \le \left \lfloor \frac{n}{2}\right \rfloor$ 
and $\cub(G) \le \left \lfloor \frac{2n}{3}\right \rfloor$. 
Upper bounds of boxicity in terms of parameters like maximum degree \cite{AdigaCOCOON} and 
tree-width \cite{Chandran2007} are known.
It was shown by Scheinerman \cite{Sch1} in 1984 that the boxicity of 
outer planar graphs is at most two. In 1986, Thomassen \cite{Thom1} proved that the boxicity of planar graphs is at most 3. 

Computation of boxicity is a notoriously hard problem. 
Even for $k=2$ or $3$, deciding whether boxicity (resp. cubicity) of a graph is at most $k$ 
is NP-complete \cite{Yan1,Krat1,Breu1998}.
Recently, Chalermsook et al. \cite{Chalermsook2013} proved that 
no polynomial time algorithm for approximating boxicity of bipartite graphs with approximation factor within $O(n^{1 - \epsilon})$ 
for any $\epsilon > 0$ is possible unless $\text{NP}=\text{ZPP}$. 
Same non-approximability holds in the case of split graphs and co-bipartite graphs too. 
Since cubicity and boxicity are equal for co-bipartite graphs, 
these hardness results extend to cubicity as well.    

Boxicity is also closely related to other dimensional parameters like poset dimension, interval dimension, threshold 
dimension, minimum chain cover number of bipartite graphs, Ferrers dimension of digraphs etc. \cite{Mahadev1995,Yan1}. 
These parameters also have $O(n^{1 - \epsilon})$ approximation hardness results 
for $\epsilon > 0$, assuming $\text{NP} \ne \text{ZPP}$. Further, 
unless $\text{NP} \subseteq \text{ZPTIME}(n^{\text{poly} \log n})$, for any $\gamma >0$ there is no $\frac{n}{2^{(\log n)^{{3}/{4}+\gamma}}}$
factor approximation algorithm for any of these problems including boxicity and cubicity \cite{Chalermsook2013} (for more 
details, see Section \ref{subsectionConseq}). 
\subsection*{Main results}
\begin{enumerate}
\item If $G$ is a graph on $n$ vertices, containing a clique of size $n-k$ or more, 
then $\boxi(G)$ and an optimal box representation of $G$ can be computed in time $n^2 2^{{O(k^2 \log k)}}$. 
\item Using the above result, we derive a polynomial time $2\left\lceil\frac{n\sqrt{\log \log n}}{\sqrt{\log n}}\right\rceil$ factor approximation algorithm
for computing boxicity and a $2\left\lceil\frac{n {(\log \log n)^{\frac{3}{2}}}}{\sqrt{\log n}}\right\rceil$ factor 
approximation algorithm for computing the cubicity. 
To our knowledge, no approximation algorithms for 
approximating boxicity and cubicity of general graphs within $o(n)$ factor were known till now.
\item The above algorithms also give us the corresponding box (resp. cube) representations. 
As a special case, this answers the question posed by Spinrad \cite{Spinrad2003} about polynomial time construction of $o(n)$ 
dimensional box representations for boxicity $2$ graphs in the affirmative.
\item As a consequence of our $o(n)$ factor approximation algorithm for boxicity, we derive
polynomial time $o(n)$ factor approximation algorithms for computing several related parameters: poset dimension,
interval dimension of finite posets, minimum chain cover of bipartite graphs,
Ferrers dimension of digraphs, and threshold dimension of split graphs and co-bipartite graphs. 
These algorithms seem to be the first $o(n)$ factor approximation algorithms known 
for each of these problems.
We note that obtaining an $o(n)$ factor approximation algorithm for poset dimension, 
is described as an open problem in Felsner et al. \cite{Felsner2015}. 
\end{enumerate}
\section{Prerequisites} \label{prereq} 
In this section, we give some basic facts necessary for the later part of this paper. 
For a vertex $v \in V$ of a graph $G$, we use $N_G(v)$ to denote the set of neighbors of $v$ in $G$. 
We use $G[S]$ to denote the induced subgraph of $G(V, E)$ on the vertex set $S \subseteq V$.  
If $I$ is an interval representation of an interval graph $G(V, E)$, we use $l_v(I)$ and $r_v(I)$ respectively to denote the 
left and right end points of the interval corresponding to $v\in V$ in $I$. 
The interval corresponding to $v$ is denoted as $\bigl[l_v(I), r_v(I)\bigr]$.
\begin{lemma}[Roberts \cite{Rob1}] \label{lmrob}
Let $G(V,$ $E)$ be any graph. For any $x\in V$, $\operatorname{box}(G) \le 1 + \operatorname{box}(G \setminus \{x\})$.
\end{lemma}
\begin{lemma}\label{lm4}
  Let $G(V, E)$ be a graph on $n$ vertices and let $A \subseteq V$. Let $G_1(V, E_1)$ be a supergraph of $G$ with 
$E_1 = E \cup \{(x,y)\mid x,y \in A,$ $x \ne y\}$. If a box representation $\mathcal{B}$ of $G$ is known, 
then in $O(n |\mathcal{B}|)$ time we can construct a box representation $\mathcal{B}_1$ of $G_1$ of dimension $2 |\mathcal{B}|$. 
In particular, $\operatorname{box}(G_1) \le 2 \operatorname{box}(G)$.
\end{lemma}
\begin{proof}
Let $\mathcal{B}$ $=\{I_1$, $I_2$, $\ldots$, $I_b\}$ be a box representation of $G$. 
For each $1 \le i \le b$, let $l_i = \displaystyle\min_{u\in V}$ $l_u(I_i)$ and $r_i = \displaystyle\max_{u\in V}$ $r_u(I_i)$. 
For $1\le i \le b$, let $I_{i_1}$ be the interval graph obtained from $I_i$ by assigning the intervals 
$$
\left[ l_v(I_{i_1}),r_v(I_{i_1})\right] = 
\begin{cases}  \left[l_i, r_{v}(I_{i})\right] & \text{if $v\in A$,}
\\
\left[ l_v(I_{i}), r_v(I_{i})\right] &\text{if $v \in V\setminus A$.}
\end{cases}
$$
 and let $I_{i_2}$ be the interval graph obtained from $I_i$ by assigning the intervals
$$
\left[ l_v(I_{i_2}),r_v(I_{i_2})\right] = 
\begin{cases}  \left[ l_{v}(I_{i}),r_i\right] & \text{if $v\in A$,}
\\
\left[ l_v(I_{i}),r_v(I_{i})\right] &\text{if $v \in V\setminus A$.}
\end{cases}
$$
It is easy to see that this construction can be done in $O(n b)$ time. 

Note that, in constructing $I_{i_1}$ and  $I_{i_2}$ we have only extended some of the intervals of $I_i$ and therefore, 
$I_{i_1}$ and  $I_{i_2}$ are supergraphs of $I_i$ and in turn of $G$. By construction, $A$ induces cliques in both  $I_{i_1}$ and  $I_{i_2}$, 
and thus they are supergraphs of $G_1$ too. 

 Now, consider $(u,v) \notin E$ with $u \in V \setminus A$, $v \in A$. Then $\exists i \in \{1,2,\ldots,b\}$ such that either $r_{v}(I_i) < l_u(I_i)$ 
or $r_u(I_i) < l_{v}(I_i)$. If $r_{v}(I_i) < l_u(I_i)$, then clearly the intervals $[l_i, r_{v}(I_i)]$ and $[l_u(I_i), r_u(I_i)]$ do not intersect 
and thus $(u,v) \notin E(I_{i_1})$. Similarly, if $r_u(I_i) < l_{v}(I_i)$, then $(u,v) \notin E(I_{i_2})$. If both $u, v \in V \setminus A$ 
and $(u,v) \notin E$, then $\exists i$ such that  $(u,v) \notin E(I_i)$ for some $1\le i\le b$ and clearly by construction, 
$(u,v) \notin E(I_{i_1})$ and  $(u,v) \notin E(I_{i_2})$.

  It follows that $G_1=\bigcap_{1 \le i \le b}{I_{i_1} \cap I_{i_2}}$ and 
  $\mathcal{B}_1$ $=\{I_{1_1}$, $I_{1_2}$, $I_{2_1}$, $I_{2_2}$, $\ldots$, $I_{b_1}$, $I_{b_2} \}$
  is a box representation of $G_1$ of dimension $2b$.   
  If $|\mathcal{B}|=\operatorname{box}(G)$ to start with, then we get $|\mathcal{B}'| \le 2 \operatorname{box}(G)$.
  Therefore, $\operatorname{box}(G_1) \le 2 \operatorname{box}(G)$. 
\qed
  \end{proof}
We know that there are at most $2^{O(nb \log n)}$ distinct $b$-dimensional box representations of a graph $G$ on $n$ vertices 
and all these can be enumerated in time $2^{O(nb \log n)}$ \cite[Proposition 1]{Adiga2}. 
In linear time, it is also possible to check whether a given graph is a unit interval graph and if so, 
generate a unit interval representation of it \cite{Booth}. Hence, a similar result holds for cubicity as well.
\begin{proposition}\label{propCount}
 Let $G(V, E)$ be a graph on $n$ vertices of boxicity (resp. cubicity) $b$. 
Then an optimal box (resp. cube) representation of $G$ can be computed in $2^{O(nb \log n)}$ time.
\end{proposition}
 If $S \subseteq V$ induces a clique in $G$, then it is easy to see that the intersection of all the intervals in $I$ 
corresponding to vertices of $S$ is nonempty. This property is referred to as the \textit{Helly property of intervals} 
and we refer to this common region of intervals as the \emph{Helly region} of the clique $S$.
\begin{definition}
 Let $G(V, E)$ be a graph in which $S \subseteq V$ induces a clique in $G$. 
Let $H(V, E')$ be an interval supergraph of $G$. Let $p$ be a point on the real line. 
If $H$ has an interval representation $I$ satisfying the following conditions:
\begin{itemize}
 \item[(1)] $p$ belongs to the Helly region of $S$ in $I$.
 \item[(2)] The end points of intervals corresponding to vertices of $V \setminus S$ are all distinct in $I$.
 \item[(3)] For each $v \in S$,\\ $l_v(I)=\min \left( p,\displaystyle\min_{u \in N_G(v) \cap (V \setminus S)} {r_u(I)} \right)$ 
and\\$r_v(I)=\max \left( p,\displaystyle \max_{u \in N_G(v)\cap (V \setminus S)} {l_u(I)}\right)$
\end{itemize}
then we call $I$ a nice interval representation of $H$ with respect to $S$ and $p$. 
If $H$ has a nice interval representation with respect to clique $S$ and some point $p$, 
then $H$ is called  a nice interval supergraph of $G$ with respect to clique $S$. 
\end{definition}
\begin{lemma}\label{lemDeriveNice}
Let $G(V, E)$ be a graph in which $S \subseteq V$ induces a clique in $G$. For every interval supergraph $I$ of $G$, we can derive a graph 
 $I'$ such that $I \supseteq I' \supseteq G$ and $I'$ a nice interval supergraph of $G$ with respect to $S$. 
\end{lemma}
\begin{proof}
Without loss of generality, we can assume that all $2|V|$ interval end points are distinct in $I$. (Otherwise, we can always alter the 
end points locally and make them distinct.) Let $p \in \mathbb{R}$ be a point belonging to the Helly region corresponding to $S$ in $I$. 
Let $I'$ be the interval graph defined by the interval assignments given below. 
$$
\left[ l_v(I'),r_v(I')\right] = 
\begin{cases}  [l_v(I),r_v(I)] & \text{if $v\in V \setminus S$,}
\\
[l'_v,r'_v] &\text{if $v \in S$.}
\end{cases}
$$
where $l'_v=\min \left( p,\displaystyle\min_{u \in N_G(v) \cap (V \setminus S)} {r_u(I)} \right)$ 
and $r'_v=\max \left( p,\displaystyle \max_{u \in N_G(v)\cap (V \setminus S)} {l_u(I)}\right)$.

We claim that $I \supseteq I' \supseteq G$.
Since for any vertex $v \in V$, the interval of $v$ in $I$ contains the interval of $v$ in $I$, we have $I \supseteq I'$.
It directly follows from the definition of $I'_i$ that $I'_i[V \setminus S]=I_i[V \setminus S]$. 
For any $(u, v) \in E(G)$, with $u \in V \setminus S$ and $v \in S$, the interval of $v$ intersects the interval of $u$ in $I_i$, 
by the definition of $[l'_v,r'_v]$. Vertices of $S$ share the common point $p$. 
Thus, $I \supseteq I' \supseteq G$. Now, from the definition of $I'$ it follows that
it is a nice interval supergraph of $G$ with respect to the clique $S$ and point $p$. 
\qed
\end{proof}
\begin{corollary}\label{corNiceBox}
If $G(V, E)$ has a box representation $\mathcal{B}$ of dimension $b$ and $S \subseteq V(G)$ induces a clique in $G$, 
then $G$ also has a box representation $\mathcal{B'}$ of the same dimension, 
in which $\forall I' \in \mathcal{B'}$, $I'$ is a nice interval supergraph of $G$ with respect to $S$.
\end{corollary}
\begin{proof}
  Let $\mathcal{B} = \{I_1$, $I_2$, $\ldots$, $I_b\}$ be a box representation of $G$. For each $1 \le i \le b$, let
  $I'_i$ be the  nice interval supergraph of $G$ with respect to $S$, derived from $I_i$, as stated in Lemma \ref{lemDeriveNice}. 
  Since, by Lemma \ref{lemDeriveNice} we have $I_i \supseteq I'_i \supseteq G$, for each $1 \le i \le b$, 
  it follows that $\mathcal{B'}= \{I'_1$, $I'_2$, $\ldots$, $I'_b\}$ is also a box representation of $G$. 
  Notice that $\mathcal{B'}$ satisfies our requirement.
  \qed
\end{proof}
\begin{lemma}\label{thmnice}
 Let $G$ be a graph on $n$ vertices, with its vertices arbitrarily labeled as $1, 2, \ldots, n$. 
If $G$ contains a clique of size $n-k$ or more, then :
 \begin{itemize}
  \item[(a)] A subset $A \subseteq V$ such that $|A| \le k$ and $G[V \setminus A]$ is a clique, can be computed in $O(n 2^k + n^2)$ time.
  \item[(b)] There are at most $2^{O(k \log k)}$ nice interval supergraphs of $G$ with respect to the clique $V \setminus A$. 
These can be enumerated in $n^2 2^{O(k \log k)}$ time.
  \item[(c)] By construction, vertices of the nice interval supergraphs obtained in (b) retain their original labels as in $G$.
\end{itemize}
\end{lemma}
\begin{proof}
\begin{itemize}
 \item[(a)] 
 We know that, if $G$ contains a clique of size $n-k$ or more, then the complement graph $\overline{G}$ has a vertex cover of size at most $k$. 
We can compute $\overline{G}$ in $O(n^2)$ time and a minimum vertex cover $A$ of $\overline{G}$ in $O(n 2^k)$ time \cite{Nie1}. 
We have $|A| \le k$ and  $G[V \setminus A]$ is a clique because $V\setminus A$ is an independent set in $\overline{G}$.
 \item[(b)] 
 Let $H$ be any nice interval supergraph of $G$ with respect to $V \setminus A$. 
Let $I$ be a nice interval representation of $H$ with respect to $V \setminus A$ and a point $p$. 
Let $P$ be the set of end points (both left and right) of the intervals corresponding to vertices of $A$ in $H$. 
Clearly $|P|=2|A|\le 2k$. 
The order of end points of vertices of $A$ in $I$ from left to right corresponds to a permutation of elements of $P$ and therefore, 
there are at most $(2k)!$ possibilities for this ordering. 
Moreover, note that the points of $P$ divide the real line into $|P|+1$ regions and that $p$ can belong to any of these regions. 
From the definition of nice interval representation, it is clear that, once the point $p$ and the end points of vertices of $A$ are fixed, 
the end points of vertices in $V\setminus A$ get automatically decided. 

Thus, to enumerate every nice interval supergraph $H$ of $G$ with respect to clique $V \setminus A$, it is enough to 
enumerate all the $(2k)!=2^{O(k \log k)}$ permutations of elements of $P$ and consider $|P|+1 \le 2k+1$ possible placements of $p$ in each of them. 
Some of these orderings may not produce an interval supergraph of $G$ though. In $O(n^2)$ time, we can check whether 
the resultant graph is an interval supergraph of $G$ and output the interval representation. 
The number of supergraphs enumerated is only $(2k+1)2^{O(k \log k)}= 2^{O(k \log k)}$.
\item[(c)]
Since vertices of $G$ are labeled initially, we just need to retain the same labeling during the definition and construction of 
nice interval supergraphs of $G$. (We have included this obvious fact in the statement of the lemma, just to give better clarity in later proofs.)  
\end{itemize}
\qed
\end{proof}
\section{Boxicity of graphs with large cliques}
One of the central ideas in this paper is the following theorem about computing the boxicity of graphs which contain very large cliques. 
Using this theorem, in Section \ref{approx} we derive $o(n)$ factor approximation algorithms for computing the boxicity and cubicity of graphs. 
\begin{theorem}\label{th1}
  Let $G$ be a graph on $n$ vertices, containing a clique of size $n-k$ or more. 
Then, $\operatorname{box}(G) \le k$ and an optimal box representation of $G$ can be found in time $n^2 2^{{O(k^2 \log k)}}$.
\end{theorem}
\begin{proof}
Let $G(V, E)$ be a graph on $n$ vertices containing a clique of size $n-k$ or more. 
We can assume that $G$ is not a complete graph; otherwise, the problem becomes trivial. 
Arbitrarily label the vertices of $G$ as $1, 2, \ldots, n$. Using part (a) of Lemma~\ref{thmnice}, we can compute in $O(n 2^k + n^2)$ time, 
$A \subseteq V$ such that $|A| \le k$ and $G[V \setminus A]$ is a clique. It is easy to infer from Lemma \ref{lmrob} that 
$\operatorname{box}(G)\le \operatorname{box}(G \setminus A)+ |A|$ $=$ $k$, since $\operatorname{box}(G\setminus A) = 0$ by definition. 

Let $\mathcal{F}$ be the family of all nice interval supergraphs of $G$ with respect to the clique $V \setminus A$. 
By Corollary \ref{corNiceBox}, if $\operatorname{box}(G)=b$, then there exists a $b$-dimensional
nice box representation of $G$, i.e., a box representation $\mathcal{B'}= \{I'_1$, $I'_2$, $\ldots$, $I'_b\}$
 of $G$ in which $I'_i \in \mathcal{F}$, for each $1 \le i \le b$.
By part (b) of Lemma \ref{thmnice}, $|\mathcal{F}|=2^{O(k \log k)}$ and all graphs in $\mathcal{F}$ can be enumerated in $n^2 2^{O(k \log k)}$ time.
Given an integer $d$, $1 \le d \le b$, verifying whether there exists a
$d$-dimensional nice box representation of $G$, and producing if one exists, can be done in $n^2 2^{{O(k \cdot d \log k)}}$ time, as follows: 
consider every subfamily $\mathcal{F}' \subseteq \mathcal{F}$ with $|\mathcal{F}'|=d$
and check if $\mathcal{F}'$ gives a valid box representation of $G$ (this validation is straightforward because vertices of supergraphs in
$\mathcal{F}'$ retain their original labels 
as explained in $G$ by part (c) of Lemma \ref{thmnice}). We might have to repeat this process  
for $1 \le d \le b$ in that order, to identify the optimum dimension $b$. Hence the total time required to compute an 
optimal box representation of $G$ is $b n^2 2^{{O(k\cdot b \log k)}}$, which is $n^2 2^{{O(k^2 \log k)}}$, 
because $b \le k$ by the first part of this theorem.
\qed
\end{proof}
\section{Approximation algorithms for computing boxicity and cubicity} \label{approx}
In this section, we use Theorem \ref{th1} and derive an $o(n)$ factor approximation algorithms for boxicity and cubicity. 
Let $G(V, E)$ be the given graph with $|V|=n$. Without loss of generality, we can assume that $G$ is connected. 
Let $k = \left \lceil \frac{\sqrt{\log n}}{\sqrt{\log \log n}}\right \rceil$ and 
$t = \left \lceil \frac{n\sqrt{\log \log n}}{\sqrt{\log n}} \right \rceil \ge \left \lceil \frac{n}{k} \right \rceil$. 
The algorithm proceeds by defining $t$ supergraphs of $G$ and computing their optimal box representations. 
Let the vertex set $V$ be partitioned arbitrarily into $t$ sets $V_1, V_2, \ldots, V_t$ where $|V_i| \le k$, for each $1 \le i \le t$. 
We define supergraphs $G_1, G_2, \ldots, G_t$ of $G$ with $G_i(V, E_i)$ defined by setting 
$E_i = E \cup \{(x,y)\mid x,y \in$ $V \setminus V_i \text{ and } x\ne y \}$, for $1\le i\le t$.
\begin{lemma}\label{th2}
 Let $G_i$ be as defined above, for $1\le i\le t$. An optimal box representation $\mathcal{B}_i$ of $G_i$ can be computed 
in $n^{O(1)}$ time, where $n=|V|$. 
\end{lemma}
\begin{proof}
Noting that $G[V\setminus V_i]$ is a clique and $|V_i| \le k=\left \lceil \frac{\sqrt{\log n}}{\sqrt{\log \log n}}\right \rceil$, 
by Theorem~\ref{th1}, we can compute an optimal box representation $\mathcal{B}_i$ of $G_i$ in $n^2 2^{{O(k^2 \log k)}}$ $=$ $n^{O(1)}$ time, where $n=|V|$. 
\qed
\end{proof}
\begin{lemma}\label{lmboxappro}
Let $\mathcal{B}_i$ be as computed above, for $1\le i\le t$. 
Then, $\mathcal{B}=\displaystyle \bigcup_{1\le i\le t}{\mathcal{B}_i}$ is a valid box representation of $G$ 
such that $|\mathcal{B}| \le t' \operatorname{box}(G)$, where $t'$ is $2\left\lceil\frac{n\sqrt{\log \log n}}{\sqrt{\log n}}\right\rceil$. 
The box representation $\mathcal{B}$ is computable in $n^{O(1)}$ time.
\end{lemma}
\begin{proof}
 We can compute optimal box representations $\mathcal{B}_i$ of $G_i$, 
for $1 \le i \le t=  \left \lceil \frac{n\sqrt{\log \log n}}{\sqrt{\log n}} \right \rceil$ 
as explained in Lemma \ref{th2} in total $n^{O(1)}$ time. Observe that $E(G)=E(G_1) \cap E(G_2) \cap \cdots \cap E(G_t)$. 
Therefore, it is a trivial observation that the union $\mathcal{B}=\displaystyle \bigcup_{1\le i\le t}{\mathcal{B}_i}$ gives us a valid 
box representation of $G$. 

We will prove that this representation gives the approximation ratio as required. By Lemma \ref{lm4} 
we have, $|\mathcal{B}_i|= \operatorname{box}(G_i) \le 2 \operatorname{box}(G)$. 
Hence, $|\mathcal{B}|= \sum_{i=1}^{t}|\mathcal{B}_i| \le 2t \operatorname{box}(G)$. 
\qed
\end{proof}
The box representation $\mathcal{B}$ obtained from Lemma \ref{lmboxappro} can be extended to a cube representation $\mathcal{C}$ of $G$ as 
stated in the following lemma. 
\begin{lemma}\label{lmcubeappro}
 A cube representation $\mathcal{C}$ of $G$, such that $|\mathcal{C}| \le t' \operatorname{cub}(G)$, where $t'$ is 
  $2 \left \lceil\frac{n {(\log \log n)}^{\frac{3}{2}}}{\sqrt{\log n}}\right\rceil$, can be computed in $n^{O(1)}$ time.
\end{lemma}
\begin{proof}
 We can compute optimal box representations $\mathcal{B}_i$ of $G_i$, for $1 \le i \le t=  \left \lceil \frac{n\sqrt{\log \log n}}{\sqrt{\log n}} \right \rceil$ as explained in Lemma \ref{th2} in $O(n^4)$ time. By \cite[Corollary 2.1]{Adiga10} we know that, from the optimal box representation $\mathcal{B}_i$ of $G_i$, in $O(n^2)$ time, we can construct a cube representation $\mathcal{C}_i$ of $G_i$ of dimension $\operatorname{box}(G_i) \lceil{\log \alpha(G_i)}\rceil$, where $\alpha(G_i)$ is the independence number of $G_i$ which is at most $|V_i|$. (Recall the assumption that $G$ is connected.) 

It is easy to see that $\mathcal{C}=\displaystyle \bigcup_{1\le i\le t}{\mathcal{C}_i}$ gives us a valid cube representation of $G$. We will prove that this cube representation gives the approximation ratio as required. We have,
\begin{eqnarray}
|\mathcal{C}| = \sum_{i=1}^{t}|\mathcal{C}_i| \le \sum_{i=1}^{t}{|\mathcal{B}_i| \lceil{\log \alpha(G_i)}\rceil}
\le \sum_{i=1}^{t}{|\mathcal{B}_i| \lceil \log k \rceil}
\le 2t \operatorname{box}(G) \log \log n
\le 2t \log \log n \operatorname{cub}(G) \nonumber
 \end{eqnarray}
\qed
\end{proof}
Combining Lemma \ref{lmboxappro} and Lemma \ref{lmcubeappro}, we get the following theorem which gives $o(n)$ factor approximation algorithms 
for computing boxicity and cubicity. 
\begin{theorem}\label{thgen}
 Let $G(V,E)$ be a graph on $n$ vertices. Then a box representation $\mathcal{B}$ of $G$, such that $|\mathcal{B}| \le t \operatorname{box}(G)$, 
where $t$ is $2\left\lceil\frac{n\sqrt{\log \log n}}{\sqrt{\log n}}\right\rceil$, can be computed in polynomial time. Further, a cube representation 
$\mathcal{C}$ of $G$, such that $|\mathcal{C}| \le t' \operatorname{cub}(G)$, where $t'$ is 
$2 \left \lceil\frac{n {(\log \log n)}^{\frac{3}{2}}}{\sqrt{\log n}}\right\rceil$, can also be computed in polynomial time.
\end{theorem}
\subsection{Consequences of Theorem \ref{thgen}}\label{subsectionConseq}
Now, we describe how Theorem \ref{thgen} can be used to derive sublinear approximation algorithms for some well-known problems
whose computational complexity is closely related to that of boxicity.
\paragraph{Chain cover of bipartite graphs}
A bipartite graph is a chain graph, if it does not contain an induced matching of size $2$. Given a bipartite graph $G(V, E)$, the minimum chain
cover number of $G$, denoted by $\ch(G)$ is the smallest number of chain graphs on the vertex set $V$ such that the union of their edge sets is $E(G)$.
It is well-known that $\ch(G)=\boxi(\overline{G})$ \cite{Yan1}.
\begin{corollary}\label{corChaincover}
There is a polynomial time $2\left\lceil\frac{n\sqrt{\log \log n}}{\sqrt{\log n}}\right\rceil$ factor approximation algorithm to compute the minimum chain cover number of an $n$-vertex bipartite graph.
\end{corollary}
\paragraph{Threshold dimension of split graphs} The concept of threshold graphs and threshold dimension was introduced by 
Chv\'{a}tal and Hammer \cite{Chvatal77} while studying some set-packing problems. 
A graph $G(V, E)$ is called a threshold graph if there exists $s \in \mathbb{R}$ and a labeling of vertices $w: V \mapsto \mathbb{R}$ 
such that $\forall u, v \in V, (u, v) \in E \Leftrightarrow w(u) + w(v) \ge s$. The threshold dimension of $G$, denoted by $t(G)$
is the minimum number of threshold subgraphs required to cover $E(G)$.
Even for split graphs, threshold dimension is hard to approximate within an $O(n^{1-\epsilon})$ factor for any $\epsilon >0$, unless
$\text{NP}=\text{ZPP}$ \cite{Chalermsook2013,Mahadev1995}. 
\begin{corollary}
There is a polynomial time $2\left\lceil\frac{n\sqrt{\log \log n}}{\sqrt{\log n}}\right\rceil$ factor approximation algorithm to compute the threshold dimension of any split graph on $n$ vertices.
\end{corollary}
\begin{proof}
 Given any split graph $G$, there is a polynomial time method to construct a bipartite graph $H$ on the same vertex set
such that $t(G)=\ch(H)$ \cite{Mahadev1995}. From the approximation algorithm for computing $\ch(H)$, the result follows.
\qed
\end{proof}
\paragraph{Threshold dimension of co-bipartite graphs}
Cozzens et al.~\cite{Cozzens1991} showed that if $G$ is a co-bipartite graph, an associated split graph $G'$ on the same vertex set can be constructed in polynomial time, such that for any $k \ge 2$, $t(G) \le k$ if and only if $t(G')\le k$. This reduction shows that the hardness result of threshold dimension of split graphs is also applicable for the threshold dimension of co-bipartite graphs. Moreover, we get the following.
\begin{corollary}
There is a polynomial time $2\left\lceil\frac{n\sqrt{\log \log n}}{\sqrt{\log n}}\right\rceil$ factor approximation algorithm to compute the threshold dimension of any co-bipartite graph on $n$ vertices.
\end{corollary}
\paragraph{Partial order dimension}
This concept was introduced by Dushnik and Miller in 1941 \cite{Dushnik41}.
A partially ordered set (poset) $\mathcal{P} = (X, P)$ consists of a nonempty set $X$ and a binary relation $P$ on $X$ that is reflexive,
antisymmetric and transitive. If every pair of distinct elements of $X$ are comparable under the relation $P$, then $(X, P)$ is called a 
total order or a linear order. A linear extension of a partial order $(X, P)$ is a linear order $(X, P')$ such that
$\forall x, y \in X$, $(x, y) \in P \Rightarrow (x, y) \in P'$. 
The dimension of a poset $\mathcal{P} = (X, P)$, denoted by $\dime(\mathcal{P})$ is defined as the smallest integer $k$ such that $\mathcal{P}$ can be expressed as the intersection of $k$ linear extensions $(X, P_1), (X, P_2), \ldots, (X, P_k)$ of $\mathcal{P}$: i.e., if $\forall x, y \in X$, $(x, y) \in P \Leftrightarrow (x, y) \in P_i$, for each $1 \le i \le k$.
 
A height-two poset is a poset $(X, P)$ in which all elements of $X$ are either minimal elements or maximal elements under the relation $P$. 
Even in the case of height-two posets, partial order dimension is hard to approximate within an $O(n^{1-\epsilon})$ factor for any $\epsilon >0$, unless
$\text{NP}=\text{ZPP}$ \cite{Chalermsook2013}. 
A height-two poset $\mathcal{P}=(X, P)$ in which $X_1$ is the set of minimal elements and $X_2$ is the set of maximal elements 
can be associated with a bipartite graph $B(\mathcal{P})$ with vertex set $X$ and edge set given by 
$\{(x, y): x \in X_1$, $y \in X_2$, $(x, y)\notin P \}$ \cite{Yan1}. 
\begin{corollary}\label{corPoset}
There is a  polynomial time $O\left(\frac{n\sqrt{\log \log n}}{\sqrt{\log n}}\right)$ factor approximation algorithm to compute the partial order dimension of a poset $\mathcal{P} = (X, P)$ defined on an $n$-element set $X$.
\end{corollary}
\begin{proof}
Let $\mathcal{P} = (X, P)$  be a poset with $|X|=n$. 
By a construction given by R. Kimble \cite{Trotter78}, given a poset $\mathcal{P} = (X, P)$ of arbitrary height, we can construct a height-two poset 
$\mathcal{P'} = (Y, P')$ from $\mathcal{P} = (X, P)$ in polynomial time so that $\dime(\mathcal{P}) \le \dime(\mathcal{P}') \le 1+ \dime(\mathcal{P})$ 
and $|Y|=2|X|$. It is also known that $\dime(\mathcal{P})=\ch(B(\mathcal{P}'))$ \cite{Yan1}.
Therefore, by computing $\ch(B(\mathcal{P}'))$ using the algorithm given by Corollary \ref{corChaincover}, we can compute a
$O\left(\frac{n\sqrt{\log \log n}}{\sqrt{\log n}}\right)$ approximation of $\dime(\mathcal{P})$.
\qed
\end{proof}
\paragraph{Interval dimension of posets}
A poset $(X, P)$ is an interval order, if each $x \in X$ can be assigned an open interval $(l_x, r_x)$ of the real line such that 
$(x, y) \in P$ if and only if $r_x \le l_y$. 
An interval order extension of a partial order $(X, P)$ is an interval order $(X, P')$ such that
$\forall x, y \in X$, $(x, y) \in P \Rightarrow (x, y) \in P'$.
The interval dimension of a poset $\mathcal{P} = (X, P)$, denoted by $\idime(\mathcal{P})$, is defined as the 
smallest integer $k$ such that $\mathcal{P}$ can be expressed as the intersection of $k$ interval order extensions 
of $\mathcal{P}$.
Since linear orders are interval orders, it follows that $\idime(\mathcal{P}) \le \dime(\mathcal{P})$. On the other hand, the poset dimension of an
interval order can be large. 

Since the height-two poset $\mathcal{P'}$ given by Kimble's construction \cite{Trotter78,Yan1} from an arbitrary finite poset $\mathcal{P}$ 
satisfies $\dime(\mathcal{P})'=\ch(B(\mathcal{P}'))$ 
and $\ch(B(\mathcal{P}'))=\idime(\mathcal{P'})$ \cite{Yan1}, from the approximation hardness of poset-dimension \cite{Chalermsook2013}, we can see that interval dimension is hard to approximate within an $O(n^{1-\epsilon})$ factor for any $\epsilon >0$, unless $\text{NP}=\text{ZPP}$. 
Felsner et al.~\cite{FelsnerDimension} showed that given a poset $(X, P)$, it is possible to construct
another poset $(Y, P')$ in polynomial time, such that $|Y|=2|X|$ and $\idime(\mathcal{P}) = \dime(\mathcal{P}')$. 
\begin{corollary}
There is a polynomial time $O\left(\frac{n\sqrt{\log \log n}}{\sqrt{\log n}}\right)$ factor approximation algorithm for 
computing the interval dimension of any poset $\mathcal{P} = (X, P)$ 
defined on a set $X$ of $n$ elements.
\end{corollary}
\paragraph{Ferrers dimension of digraphs} 
Ferrers relations were introduced by Riguet in 1950's \cite{Riguet}. 
A digraph $G(V, E)$ is called a Ferrers digraph when there exists a linear
order $(V, L)$ such that, for every $x, y, z \in V$, if $(x, y) \in L$ and $(y, z) \in E$ then $(x, z) \in E$. 
The Ferrers dimension \cite{CogisDimension} of a digraph $G$ is the smallest number of Ferrers digraphs whose intersection is $G$. 
Since a partial order $\mathcal{P}$ has $\dime(\mathcal{P})$ equal to the Ferrers dimension of its underlying digraph \cite{CogisDimension}, 
Ferrers dimension is also hard to approximate within an $O(n^{1-\epsilon})$ factor for any $\epsilon >0$, unless $\text{NP}=\text{ZPP}$. 
Cogis \cite{CogisDimension} showed that given a digraph $G(V, E)$, a poset $\mathcal{P} = (X, P)$ can be constructed in polynomial time, 
such that $|X| \le 2|V|$ and the poset dimension of $\mathcal{P}$ is equal to the Ferrers dimension of $G$. 
\begin{corollary}
There is a polynomial time $O\left(\frac{n\sqrt{\log \log n}}{\sqrt{\log n}}\right)$ factor approximation algorithm 
for computing the Ferrers dimension of a digraph on $n$ vertices.
\end{corollary}
\section{Conclusion}
We have presented $o(n)$ factor
approximation algorithms for computing the boxicity and cubicity of graphs. 
Using these algorithms, we also derived $o(n)$ factor approximation algorithms for 
some related well-known problems, including poset dimension and 
Ferrers dimension. 
To the best of our knowledge, for none of these problems 
polynomial time sublinear factor approximation algorithms were known previously.
Since polynomial time approximations within an $O(n^{1-\epsilon})$ factor for any $\epsilon >0$ is considered unlikely for any of these problems, 
no significant improvement in the approximation factor can be expected. 

\end{document}